\documentclass[copyright,creativecommons]{eptcs}

\usepackage{breakurl}

\usepackage{proof}
\usepackage{subfigure}
\usepackage{epsf}
\usepackage{graphics}
\usepackage{wrapfig}
\usepackage{mathrsfs}
\usepackage{url}
\usepackage{calc}
\usepackage{amsmath,amssymb,amsfonts,mathrsfs,latexsym,stmaryrd}
\usepackage[all]{xy}
\newcommand{\procone}{P}
\newcommand{\proctwo}{Q}
\newcommand{\procthree}{R}
\newcommand{\procfour}{S}
\newcommand{\procfive}{T}

\newcommand{\varone}{x}
\newcommand{\vartwo}{y}
\newcommand{\varthree}{z}

\newcommand{\valueone}{V}
\newcommand{\valuetwo}{W}
\newcommand{\valuethree}{Z}
\newcommand{\chanone}{a}
\newcommand{\chantwo}{b}
\newcommand{\chanthree}{c}
\newcommand{\conone}{\Gamma}
\newcommand{\contwo}{\Delta}
\newcommand{\conthree}{\Lambda}
\newcommand{\confour}{\Theta}
\newcommand{\emcon}{\emptyset}
\newcommand{\emproc}{\textbf{0}}
\newcommand{\tdone}{\pi}
\newcommand{\tdtwo}{\rho}
\newcommand{\tdthree}{\sigma}
\newcommand{\inp}[3]{#1( #2).#3}
\newcommand{\einp}[3]{#1(!#2).#3}
\newcommand{\sinp}[3]{#1(\Box#2).#3}
\newcommand{\out}[3]{\overline{#1}\langle #2 \rangle.#3}
\newcommand{\outwc}[2]{\overline{#1}\langle #2 \rangle}
\newcommand{\para}[2]{#1\;||\;#2}
\newcommand{\rest}[2]{(\nu #1)#2}
\newcommand{\abstr}[2]{\lambda #1.#2}
\newcommand{\eabstr}[2]{\lambda\! !#1.#2}
\newcommand{\sabstr}[2]{\lambda\! \Box#1.#2}
\newcommand{\app}[2]{#1#2}
\newcommand{\bang}[1]{!#1}
\newcommand{\sang}[1]{\Box#1}
\newcommand{\dang}[1]{\Diamond#1}
\newcommand{\dies}[1]{\# #1}
\newcommand{\nfo}[2]{\mathbb{FO}(#1,#2)}
\newcommand{\bd}[1]{\mathbb{B}(#1)}

\newcommand{\df}[1]{\mathbb{D}(#1)}
\newcommand{\weip}[2]{\mathbb{W}_{#2}(#1)}
\newcommand{\wei}[1]{\mathbb{W}(#1)}
\newcommand{\webip}[2]{\mathbb{I}_{#2}(#1)}
\newcommand{\webi}[1]{\mathbb{I}(#1)}
\newcommand{\pgrp}[2]{\mathbb{P}_{#2}(#1)}
\newcommand{\pgr}[1]{\mathbb{P}(#1)}

\newcommand{\subst}[3]{#1[#2/#3]}
\newcommand{\os}[2]{#1\rightarrow_\mathsf{P} #2}
\newcommand{\osrel}{\rightarrow_\mathsf{P}}
\newcommand{\osl}[2]{#1\rightarrow_\mathsf{L} #2}
\newcommand{\oslp}[3]{#1\rightarrow_\mathsf{L}^{#2} #3}
\newcommand{\wfjv}[2]{#1\vdash_\mathsf{V} #2}
\newcommand{\wfjp}[2]{#1\vdash_\mathsf{P} #2}
\newcommand{\wfjsv}[2]{#1\vdash_\mathsf{SV} #2}
\newcommand{\wfjsp}[2]{#1\vdash_\mathsf{SP} #2}
\newcommand{\wfjev}[2]{#1\vdash_\mathsf{EV} #2}
\newcommand{\wfjep}[2]{#1\vdash_\mathsf{EP} #2}
\newcommand{\congr}[2]{#1\equiv #2}
\newcommand{\congrs}{\equiv}
\newcommand{\embv}[1]{[#1]_\mathsf{V}}
\newcommand{\embp}[1]{[#1]_\mathsf{P}}

\newcommand{\server}{\mathit{SERVER}}
\newcommand{\lserver}{\mathit{SERVER}_!}
\newcommand{\elserver}{\mathit{SERVER}_\Box}
\newcommand{\compon}{\mathit{COMP}}
\newcommand{\lcompon}{\mathit{COMP}_!}
\newcommand{\elcompon}{\mathit{COMP}_\Box}
\newcommand{\procomega}{\mathit{OMEGA}}
\newcommand{\lprocomega}{\mathit{OMEGA}_!}
\newcommand{\procdelta}{\mathit{DELTA}}
\newcommand{\lprocdelta}{\mathit{DELTA}_!}

\newcommand{\chans}{\mathscr{C}}
\newcommand{\inpchans}{\mathscr{IC}}

\newcommand{\size}[1]{|#1|}

\newcommand{\hopi}{$\mathbf{HO}\pi$}
\newcommand{\hopisw}{$\mathbf{HO}\pi^{\mathtt{unit},\rightarrow,\diamond}$}
\newcommand{\lhopi}{$\mathbf{LHO}\pi$}
\newcommand{\shopi}{$\mathbf{SHO}\pi$}
\newcommand{\eshopi}[1]{$\mathbf{EHO}\pi(#1)$}

\newcommand{\SLL}{\textbf{SLL}}

\newtheorem{lemma}{Lemma}
\newtheorem{proposition}{Proposition}
\newtheorem{theorem}{Theorem}

\newenvironment{proof}{\begin{trivlist}
       \item[\hskip \labelsep {\bfseries Proof.}]}{\hfill $\Box$ \end{trivlist}}

\newenvironment{varitemize}
{
\begin{list}{\labelitemi}
{\setlength{\itemsep}{0.0mm}
 \setlength{\topsep}{0.0mm}
 \setlength{\parindent}{0.0mm}
 \setlength{\parskip}{0.0mm}
 \setlength{\parsep}{0.0mm}
 \setlength{\partopsep}{0.0mm}
 \setlength{\leftmargin}{15pt}
 \setlength{\labelsep}{5pt}
 \setlength{\labelwidth}{10pt}}}
{
 \end{list} 
}

\newcounter{number}

\newenvironment{varenumerate}
{\begin{list}{\arabic{number}.}
  {
   \usecounter{number}
   \setlength{\labelwidth}{4.0mm}
   \setlength{\labelsep}{2.0mm}
   \setlength{\itemindent}{0.0mm}
   \setlength{\itemsep}{0.0mm}
   \setlength{\topsep}{0.0mm}
   \setlength{\parskip}{0.0mm}
   \setlength{\parsep}{0.0mm}
   \setlength{\partopsep}{0.0mm}
  }
}
{\end{list}}

\title{Light Logics and Higher-Order Processes}
\author{
Ugo Dal Lago
\institute{Universit\`a di Bologna \\ INRIA Sophia Antipolis}
\email{dallago@cs.unibo.it}
\and
Simone Martini
\institute{Universit\`a di Bologna \\ INRIA Sophia Antipolis}
\email{martini@cs.unibo.it}
\and
Davide Sangiorgi
\institute{Universit\`a di Bologna \\ INRIA Sophia Antipolis}
\email{sangio@cs.unibo.it}}

\def\titlerunning{Light Logics and Higher-Order Processes}
\def\authorrunning{U. Dal Lago \& S. Martini \& D. Sangiorgi}

\begin{document}
\maketitle
\begin{abstract}
\noindent
We show that the techniques for resource control that have been developed in
the so-called ``light logics'' can be fruitfully applied also
to process algebras. In particular, we present a restriction
of Higher-Order $\pi$-calculus inspired by Soft Linear Logic.
We prove that any soft process terminates in polynomial time. We argue that the class of soft processes may be naturally
enlarged so that interesting processes are expressible, still maintaining
the polynomial bound on executions.
\end{abstract}
\section{Introduction}
A term terminates if all its reduction sequences are of finite length.
As far as programming languages are concerned, termination means that
computation in programs will eventually stop.  In computer science,
termination has been extensively investigated in sequential languages,
where strong normalization is a synonym more commonly used.

Termination  is however interesting  also in concurrency.  
While large concurrent systems often are  supposed to run forever
(e.g., an operating system, or the Internet itself), single components are
usually expected to terminate.  For instance, if we query a
server, we may want to know that the server does not go on forever
trying to compute an answer.  Similarly, when we load an applet we
would like to know that the applet will not run forever on our
machine, possibly absorbing all the computing resources.  In general,
if the lifetime of a process can be infinite, we may want to know that
the process does not remain alive simply because of nonterminating
internal activity, and that, therefore, the process will eventually
accept interactions with the environment.  

Another motivation for studying termination in concurrency is to
exploit it within techniques aimed at guaranteeing properties such as
responsiveness and lock-freedom \cite{KobayashiS08}, which
{intuitively} indicate that certain communications or synchronizations
will eventually succeed (possibly under some fairness assumption). 
In message-passing languages such  as those in the
$\pi$-calculus family (Join Calculus, Higher-Order $\pi$-calculus, Asynchronous
$\pi$-calculus, etc.) 
most liveness properties can be reduced to instances of lock-freedom.
Examples, in  a client-server system, 
are the liveness properties that 
 a client  request will 
 eventually be received   by  the
 server, or that  a 
 server, once accepted a request, will  eventually send back an answer.

However, termination alone may not be satisfactory.  If a query to a
server produces a computation that terminates after a very long time,
from the client point of view this may be the same as a
nonterminating (or failed) computation.  Similarly, an applet loaded
on our machine that starts a very long computation, may engender an
unacceptable consumption of local resources, and may possibly be
considered a ``denial of service'' attack.  In other words, without
precise bounds on the time to complete a computation, termination may 
be indistinguishable from nontermination.

Type disciplines are among the most general techniques to ensure termination of programs.
Both in the sequential and in the concurrent case, type systems have been
designed to characterize classes of terminating programs. 
It is interesting that, from the fact that a program has a type, we may often extract information on the structure
of the program itself (e.g., for the simple types, the program has no
self applications). If termination (or, more generally, some property of
the computation) is the main interest, it is only this structure that matters,
and not the specifics of the types. In this paper we take
this perspective, and apply to a certain class of programs (Higher-Order
$\pi$-calculus terms) the structural restrictions suggested by the types of Soft
Linear Logic~\cite{Lafont:TCS01}, a fragment of Linear Logic~\cite{Gir87} 
characterizing polynomial time computations. 

Essential contribution of Linear Logic has been the \emph{refinement} it allows 
on the analysis of computation. The (previously 
atomic) step of function application is decomposed into a duplication phase
(during which the argument is duplicated the exact number of times it will
be needed during the computation), followed by the application of a \emph{linear}
function (which will use each argument exactly once). The emphasis here
is not on restricting the class of programs---in many cases, any traditional
program (e.g., any $\lambda$-term, even a divergent one) could be annotated 
with suitable \emph{scope information} (\emph{boxes}, in the
jargon) in such a way that the annotated program behaves as the original one.
However, the new annotations embed information on
the computational behavior that was unexpressed (and inexpressible) before.
In particular, boxes delimit those parts of data that will be (or may be) duplicated or erased
during computation.

It is at this stage that one may apply \emph{restrictions}. By building on the 
scopes exposed in the new syntax, 
we may restrict the computational behavior of a term. In the sequential case
several achievements have been obtained via the so-called 
\emph{light logics}~\cite{girard98light,AspertiRoversi:TOCL02,Lafont:TCS01}, which allow for type systems for $\lambda$-calculus
exactly characterizing several complexity classes (notably,
elementary time, polynomial type, polynomial space, logarithmic space).
This is obtained by limitations on the way the scopes (boxes) may be manipulated.
For the larger complexity classes (e.g., elementary time) one forbids that 
during computation  one
scope may enter inside another scope (their nesting depth
remains constant). For smaller classes (e.g., polynomial time) one also forbids
that a duplicating computation could drive another duplication.
The exact way this is obtained depends on the particular discipline (either
\`a la Light Linear Logic, or \`a la Soft Linear Logic).

The aim of this paper is to apply for the first time
these technologies to the concurrent case, in particular
to Higher-Order $\pi$-calculus~\cite{SangiorgiWalker}. We closely follow the
pattern we have delineated above. First, we introduce (higher-order) processes,
which we then annotate with explicit scopes, where the new construct ``!'' 
marks duplicable entities. This is indeed a refinement, and not a restriction --- 
any process in the first calculus may be simulated by an annotated one. 
We then introduce our main object of study --- annotated processes restricted 
with the techniques of Soft Linear Logic. We show that the number of internal
actions performed by processes of this calculus is polynomially bounded (Section~\ref{sect:shopi}), a
property that we call \emph{feasible termination}. Moreover,
an extension of the calculus capturing a natural example will be presented (Section~\ref{sect:eshopi}).

We stress that we used in the paper a pragmatic approach --- take from the logical side
tools and techniques that may be suitable to obtain general bounds
on the computing time of processes. We are not looking for
a general relation between logical systems and process algebras that could
realize a form of Curry-Howard correspondence among the two. That would be
a much more ambitious goal, for which other techniques --- and different
success criteria --- should be used.

\paragraph{Related Work}
A number of works have recently studied type systems that ensure
termination in mobile
processes, e.g.\  \cite{YBH11,DemangeonHS09,DemangeonHS10}.
They are  quite different from the present paper.
First, the techniques employed are measure-based techniques, or
logical relations, or combinations of these, rather than techniques
inspired by linear logics, as done here. Secondly, the objective is pure
termination, whereas here we aim at deriving polynomial bounds on the number of
steps that lead to termination. (In some of the measure-based  systems
bounds can actually be derived, but they are usually exponential
with respect to
integer annotations that appear in the types.) Thirdly,
with the exception of \cite{DemangeonHS10}, all works analyse
name-passing calculi such as the $\pi$-calculus, whereas here we consider
higher-order calculi in which terms of the calculus
are exchanged  instead of names.

Linear Logic has been applied to mobile processes by
Ehrhard and Laurent \cite{EL10}, who have studied encodings of
$\pi$-calculus-like languages into Differential Interaction Nets
\cite{EhrhardR06},  an
extension of the Multiplicative Exponential fragment of Linear Logic.
The encodings are meant to be tests for the expressiveness of
Differential Interaction Nets; the issue of termination does not
arise, as the process calculi encoded are finitary.
Amadio and Dabrowski \cite{AmadioD07}  have applied ideas from term rewriting
to a $\pi$-calculus enriched with synchronous constructs \`a la
Esterel. Computation in processes proceeds synchronously,
divided into cycles called instants.
A static analysis and a
finite-control condition
 guarantee
that, during each instant, the size of a program and
the times it takes to complete the instant
are   polynomial on the size of the program and
the input values at the beginning of the
instant.

\section{Higher-Order Processes}\label{sect:hopi}
This section introduces the syntax and the operational semantics of
processes. We call \hopi\ the calculus of processes we are going to define
(it is  the calculus \hopisw\ in~\cite{SangiorgiWalker}).
In \hopi\ the values exchanged in interactions can be first-order
values and
higher-order values, i.e., terms containing processes. 
For economy, the only first-order value  employed is the unit value
$\star$, and the only higher-order values are parametrised processes,
called abstractions (thus we forbid direct communication of processes;
to communicate a process we must add a dummy parameter to it).
The process constructs are nil, parallel composition, input, output, 
restriction, and application.  Application is the destructor for
abstraction: it allows us to instantiate  the formal  parameters of an
abstraction.
Here is the complete grammar:
\begin{align*}
\procone&::=\emproc\mid\para{\procone}{\procone}\mid\inp{\chanone}{\varone}{\procone}\mid
  \out{\chanone}{\valueone}{\procone}\mid\rest{\chanone}{\procone}\mid\app{\valueone}{\valueone};\\
\valueone&::= \star\mid \varone\mid \abstr{\varone}{\procone};
\end{align*}
where $\chanone$ ranges over a denumerable set $\chans$ of channels, and
$ \varone$ over the  denumerable set  of variables.
Input, restriction, and abstractions are binding constructs, and give
rise in the expected way to the notions of free and bound channels and
of free and bound
variables, as well as of $\alpha$-conversion.

Ill-formed terms such as
$\app{\star}{\star}$ can be avoided by means of a type systems. 
The details are standard and are omitted here;
see~\cite{SangiorgiWalker}. 

The operational semantics, in the reduction style, is presented in
Figure~\ref{fig:oshopi}, and uses the auxiliary relation of 
\emph{structural congruence}, written $\equiv$.  This is the smallest
congruence closed under the following rules:
\begin{align*}
\procone&\congrs\proctwo\mbox{ if $\procone$ and $\proctwo$ are $\alpha$-equivalent};\\
\para{\procone}{(\para{\proctwo}{\procthree})}&\congrs\para{(\para{\procone}{\proctwo})}{\procthree};\\
\para{\procone}{\proctwo}&\congrs\para{\proctwo}{\procone};\\
\rest{\chanone}{(\rest{\chantwo}{\procone})}&\congrs\rest{\chantwo}{(\rest{\chanone}{\procone})};\\
(\rest{\chanone}{\para{\procone}{\proctwo}})&\congrs\para{(\rest{\chanone}{\procone})}{\proctwo}\mbox{
  if  $\chanone$ is not free in ${\proctwo}$};
\end{align*}
Unlike other presentations of structural congruence, we disallow the
garbage-collection laws
$\para{\procone}{\emproc}\congrs\procone$ and
$\rest{\chanone}{\emproc}\congrs\chanone$, which are troublesome for
our resource-sensitive analysis.
The reduction relation is written $\osrel$, and is defined on
processes without free variables. 

\begin{figure}
\begin{center}
\fbox{\begin{minipage}[c]{.97\textwidth}
\vspace{12pt}
$$
\begin{array}{ccc}
\infer
{\os{\para{\out{\chanone}{\valueone}{\procone}}{\inp{\chanone}{\varone}{\proctwo}}}{\para{\procone}{\subst{\proctwo}{\varone}{\valueone}}}}
{}
&
\hspace{10pt}
&
\infer
{\os{\app{(\abstr{\varone}{\procone})}{\valueone}}{\subst{\procone}{\varone}{\valueone}}}
{}
\end{array}
$$
\vspace{2pt}
$$
\begin{array}{ccccc}
\infer
{\os{\para{\procone}{\procthree}}{\para{\proctwo}{\procthree}}}
{\os{\procone}{\proctwo}}
&
\hspace{10pt}
&
\infer
{\os{\rest{\chanone}{\procone}}{\rest{\chanone}{\proctwo}}}
{\os{\procone}{\proctwo}}
&
\hspace{10pt}
&
\infer
{\os{\procone}{\procfour}}
{ \congr{\procone}{\proctwo} & \os{\proctwo}{\procthree} & \congr{\procthree}{\procfour} }
\end{array}
$$
\end{minipage}}
\end{center}
\caption{The operational semantics of \hopi\ processes.}\label{fig:oshopi}
\end{figure}

In general, the
relation $\osrel$ is nonterminating. The prototypical example
of a nonterminating process is the following process $\procomega$:
$$ 
\procomega=\rest{\chanone}{(\para{\app\procdelta \star}{\outwc{\chanone}{\procdelta}})},
\qquad\mbox{where}\qquad\procdelta=
\abstr{\vartwo}{(\inp{\chanone}\varone{(\para{{\app{\varone\:}{\,\star}}}{\outwc{\chanone}{\varone}})})}.
$$
Indeed, it holds that $\procomega\osrel^2\procomega$.
Variants of the construction  employed for $\procomega$ can be used to show
that process  recursion can be modelled in \hopi. 
An example of this construction is  the following  $\server$ process.
It accepts a request $\vartwo$ on channel $\chantwo$ and forwards it
along $\chanthree$. After that, it can handle another
request from $\chantwo$. In contrast to $\procomega$,  $\server$ is terminating, because
there is no infinite reduction sequence starting from $\server$. Yet hand, the number of
requests $\server$ can handle is unlimited, i.e., $\server$ can be
engaged in an infinite sequence of interactions 
with its environment.
\begin{align*}
\server&=\rest{\chanone}{(\para{\app \compon\star}{\outwc{\chanone}{\compon}})};\\
\compon&=
\abstr \varthree {( 
\inp{\chanone}{\varone}{(\para{\inp{\chantwo}{\vartwo}
    {\out{\chanthree}{\vartwo}{\outwc{\chanone}{\varone}}}} 
    {\app{\varone}{\star}})})}.
\end{align*}

A remark on notation: in this paper, $!$ is the Linear Logic operator
(more precisely, an operator derived from  Linear Logic), and should
not be confused with  the replication operator often used in process
calculi such as the $\pi$-calculus. 

\section{Linearizing Processes}\label{sect:lhopi}
Linear Logic can be seen as a way to decompose the type of functions
$A\rightarrow B$ into a refined type $!A\multimap B$.
Since the argument (in $A$) may be used several (or zero)
times to compute the result in $B$, we first turn the input 
into a  duplicable
(and erasable) object (of type $!A$). We now duplicate (or erase) it the number of times
it is needed, and finally we use each of the copies exactly once to
obtain the result (this is the linear function space $\multimap$). The richer language
of types (with the new constructors $!$ and $\multimap$) is matched by new term constructs, whose goal is to 
explicitly enclose in marked scopes (boxes) those subterms that may be erased or duplicated. In the computational process
we described above, there are three main ingredients: (i) the mark on a duplicable/erasable entity; (ii) its actual duplication/erasure; (iii) the linear use of the copies. 
For reasons that cannot be discussed here (see Wadler's~\cite{Wad93} for the
notation we will use) we may adopt
a syntax where the second step (duplication) is not made fully explicit (thus resulting
in a simpler language), and where the crucial distinction is made
between linear functions (denoted by the usual syntax $\abstr{\varone}{\procone}$ --- but
interpreted in a strictly linear way: $x$ occurs once in $P$), and nonlinear functions, denoted with
$\eabstr{\varone}{\procone}$, where the $x$ may occur several (or zero) times in $\procone$.
When a nonlinear function is applied, its actual argument will be duplicated or erased. 
We enclose the argument in a box to record this fact, using an eponymous unary operator
$!$ also on terms. Since we want to control the computational behavior of duplicable entities,
a term in a !-box is protected and cannot be reduced. Only when it will be fed to
a (nonlinear) function, and thus (transparently) duplicated, its box will be opened (the mark ! disappears)
and the content will be reduced.

The constructs on \emph{terms} arising from Linear Logic have a natural counterpart in higher-order
processes, where communication and abstraction play a similar role.
This section introduces a linearization of \hopi, that we here
dub \lhopi. The grammars of processes and values are as follows:
\begin{align*}
\procone&::=\emproc\mid\para{\procone}{\procone}\mid\inp{\chanone}{\varone}{\procone}\mid\einp{\chanone}{\varone}{\procone}\mid
  \out{\chanone}{\valueone}{\procone}\mid\rest{\chanone}{\procone}\mid\app{\valueone}{\valueone};\\
\valueone&::=\star\mid\varone\mid\abstr{\varone}{\procone}\mid\eabstr{\varone}{\procone}\mid\;\bang{\valueone}.
\end{align*}
On top of the grammar, we must enforce the linearity constraints, which are expressed by the rules in Figure~\ref{fig:lhopiwff}.
\begin{figure}
\begin{center}
\fbox{\begin{minipage}[c]{.97\textwidth}
\vspace{12pt}
$$
\begin{array}{ccccc}
\infer
{\wfjp{\bang{\conone}}{\emproc}}
{}
&
\hspace{10pt}
&
\infer
{\wfjp{\conone,\contwo,\bang{\conthree}}{\para{\procone}{\proctwo}}}
{\wfjp{\conone,\bang{\conthree}}{\procone} & \wfjp{\contwo,\bang{\conthree}}{\proctwo}}
&
\hspace{10pt}
&
\infer
{\wfjp{\conone}{\inp{\chanone}{\varone}{\procone}}}
{\wfjp{\conone,\varone}{\procone}}
\end{array}
$$
\vspace{2pt}
$$
\begin{array}{ccccc}
\infer
{\wfjp{\conone}{\einp{\chanone}{\varone}{\procone}}}
{\wfjp{\conone,\bang{\varone}}{\procone}}
&
\hspace{10pt}
&
\infer
{\wfjp{\conone,\contwo,\bang{\conthree}}{\out{\chanone}{\valueone}{\procone}}}
{\wfjv{\conone,\bang{\conthree}}{\valueone} & \wfjp{\contwo,\bang{\conthree}}{\procone}}
&
\hspace{10pt}
&
\infer
{\wfjp{\conone}{\rest{\chanone}{\procone}}}
{\wfjp{\conone}{\procone}}
\end{array}
$$
\vspace{2pt}
$$
\begin{array}{ccccc}
\infer
{\wfjp{\conone,\contwo,\bang{\conthree}}{\app{\valueone}{\valuetwo}}}
{\wfjv{\conone,\bang{\conthree}}{\valueone} & \wfjv{\contwo,\bang{\conthree}}{\valuetwo}}
&
\hspace{10pt}
&
\infer
{\wfjv{\bang{\conone}}{\star}}
{}
&
\hspace{10pt}
&
\infer
{\wfjv{\bang{\conone},\varone}{\varone}}
{}
\end{array}
$$
\vspace{2pt}
$$
\begin{array}{ccccccc}
\infer
{\wfjv{\bang{\conone},\bang{\varone}}{\varone}}
{}
&
\hspace{10pt}
&
\infer
{\wfjv{\conone}{\abstr{\varone}{\procone}}}
{\wfjp{\conone,\varone}{\procone}}
&
\hspace{10pt}
&
\infer
{\wfjv{\conone}{\eabstr{\varone}{\procone}}}
{\wfjp{\conone,\bang{\varone}}{\procone}}
&
\hspace{10pt}
&
\infer
{\wfjv{\bang{\conone}}{\bang{\valueone}}}
{\wfjv{\bang{\conone}}{\valueone}}
\end{array}
$$
\vspace{2pt}
\end{minipage}}
\end{center}
\caption{Processes and values in \lhopi.}\label{fig:lhopiwff}
\end{figure}
They prove judgements in the form
$\wfjp{\conone}{\procone}$ and $\wfjv{\conone}{\valueone}$, where
$\conone$ is a \emph{context} consisting of a finite set of variables --- a single 
variable may appear in $\conone$ either as $x$ or as $!x$, but not both. 
Examples of contexts are $\varone,\bang{\vartwo}$; or $\varone,\vartwo,\varthree$; or the empty context $\emcon$. 
As usual, we write $!\conone$ when all variables of the context (if any) are !-marked. A process $\procone$ (respectively, a value $\valueone$) is \emph{well-formed} iff there is a context $\conone$ such that
$\wfjp{\conone}{\procone}$ (respectively, $\wfjv{\conone}{\valueone}$). 
In the rules with two premises, observe the implicit contractions on $!$-marked variables in the 
context --- they allow for transparent duplication. The \emph{depth} of a (occurrence of a) variable $\varone$ in a process
or value is the number of instances of the $!$ operator it is enclosed to. As an example, if $\procone=\app{(\bang{\varone})}{(\vartwo)}$,
then $\varone$ has depth $1$, while $\vartwo$ has depth $0$.

A judgement $\wfjp{\conone}{\procone}$ can informally be interpreted as follows.
Any variable  appearing as $\varone$ in $\conone$ must occur free exactly
once in $\procone$; moreover the only occurrence of $\varone$ is at depth
$0$ in $\procone$ (that is, it is not in the scope of any !). On the other hand, any variable
$\vartwo$ appearing as $\bang{\vartwo}$ in $\conone$  may occur free any
number of times in $\procone$, at any depth.
Variables like $\varone$ are  \emph{linear}, while those like $\vartwo$ are
\emph{nonlinear}. Nonlinear variables may only be bound by nonlinear 
binders (which have a $!$ to recall this fact). 

The operational semantics of \lhopi\ is a slight
variation on the one of \hopi, and can be found in Figure~\ref{fig:oslhopi}.
\begin{figure}
\begin{center}
\fbox{\begin{minipage}[c]{.97\textwidth}
\vspace{12pt}
$$
\begin{array}{ccc}
\infer
{\osl{\para{\out{\chanone}{\valueone}{\procone}}{\inp{\chanone}{\varone}{\proctwo}}}{\para{\procone}{\subst{\proctwo}{\varone}{\valueone}}}}
{}
&
\hspace{10pt}
&
\infer
{\osl{\app{(\abstr{\varone}{\procone})}{\valueone}}{\subst{\procone}{\varone}{\valueone}}}
{}
\end{array}
$$
\vspace{2pt}
$$
\begin{array}{ccc}
\infer
{\osl{\para{\out{\chanone}{\bang{\valueone}}{\procone}}{\einp{\chanone}{\varone}{\proctwo}}}
{\para{\procone}{\subst{\proctwo}{\varone}{\valueone}}}}
{}
&
\hspace{10pt}
&
\infer
{\osl{\app{(\eabstr{\varone}{\procone})}{\bang{\valueone}}}{\subst{\procone}{\varone}{\valueone}}}
{}
\end{array}
$$
\vspace{2pt}
$$
\begin{array}{ccccc}
\infer
{\osl{\para{\procone}{\procthree}}{\para{\proctwo}{\procthree}}}
{\osl{\procone}{\proctwo}}
&
\hspace{10pt}
&
\infer
{\osl{\rest{\chanone}{\procone}}{\rest{\chanone}{\proctwo}}}
{\osl{\procone}{\proctwo}}
&
\hspace{10pt}
&
\infer
{\osl{\procone}{\procfour}}
{ \congr{\procone}{\proctwo} & \osl{\proctwo}{\procthree} & \congr{\procthree}{\procfour} }
\end{array}
$$
\end{minipage}}
\end{center}
\caption{The operational semantics of \lhopi\ processes.}\label{fig:oslhopi}
\end{figure}
The two versions of communication and abstraction (i.e., the linear and the
nonlinear one) are governed by two distinct rules. In the nonlinear
case the argument to the function (or the value sent through a channel)
must be in the correct duplicable form $\bang{\valueone}$. Well-formation is preserved
by reduction:
\begin{lemma}[Subject Reduction]
If $\wfjp{\;}{\procone}$ and $\osl{\procone}{\proctwo}$, then $\wfjp{\;}{\proctwo}$.
\end{lemma}
\subsection{Embedding Processes into Linear Processes}
Processes (and values) can be embedded into linear processes (and values) as follows:
\begin{align*}
\embv{\star}&=\star; & 
\embv{\abstr{\varone}{\proctwo}}&=\eabstr{\varone}{\embp{\procone}};\\
\embp{\emproc}&=\emproc; &
\embv{\varone}&=\varone;\\
\embp{\para{\procone}{\proctwo}}&=\para{\embp{\procone}}{\embp{\proctwo}}; & 
\embp{\inp{\chanone}{\varone}{\procone}}&=\einp{\chanone}{\varone}{\embp{\procone}};\\
\embp{\out{\chanone}{\valueone}{\procone}}&=\out{\chanone}{\bang{\embv{\valueone}}}{\embp{\procone}}; & 
\embp{\rest{\chanone}{\procone}}&=\rest{\chanone}{\embp{\procone}};\\
\embp{\app{\valueone}{\valuetwo}}&=\app{\embv{\valueone}}{\bang{\embv{\valuetwo}}}.
\end{align*}
Linear abstractions and linear inputs never appear in processes
obtained via $\embp{\cdot}$: whenever a value is sent through a channel or passed to
a function, it is made duplicable. The embedding induces a simulation of processes by linear processes:
\begin{proposition}[Simulation]\label{prop:simulation}
For every process $\procone$, $\embp{\procone}$ is well-formed. Moreover,
if $\os{\procone}{\proctwo}$, then $\osl{\embp{\procone}}{\embp{\proctwo}}$.
\end{proposition}
By applying the map $\embp{\cdot}$ to our example process, $\server$, a linear process
$\lserver$ can be obtained:
\begin{align*}
\lserver&=\rest{\chanone}{(\para{\app \lcompon{(\bang{\star})}}{\outwc{\chanone}{\bang{\lcompon}}})};\\
\lcompon&=\eabstr{\varthree}{( 
\einp{\chanone}{\varone}{(\para{\einp{\chantwo}{\vartwo}
    {\out{\chanthree}{\bang{\vartwo}} {\outwc{\chanone}{\bang{\varone}}}}}{\app{\varone}{(\bang{\star})}})})}.
\end{align*}
\section{Termination in Bounded Time: Soft Processes}\label{sect:shopi}
In view of Proposition~\ref{prop:simulation}, \lhopi\ admits non terminating processes. Indeed, 
the prototypical divergent process from Section~\ref{sect:hopi} can be translated into a
linear process:
$$
\lprocomega=\rest{\chanone}{(\para{(\app\lprocdelta{(\bang{\star})})}{\outwc{\chanone}{\bang{\lprocdelta}}})},
\qquad\mbox{where}\qquad\lprocdelta =
\eabstr{\vartwo}{(
\einp{\chanone}\varone{(\para{{\app{\varone\:}{(\bang{\star})}}}{\outwc{\chanone}{\bang{\varone}}})})}.
$$
$\lprocomega$ cannot be terminating, since $\procomega$ itself does not terminate.

The more expressive syntax, however, may reveal
\emph{why} a process does not terminate. If we trace its execution, we see
that the divergence of $\lprocomega$ comes from $\lprocdelta$, where
 $\varone$ appears free twice in the inner body
${(\para{{\app{\varone\:}{(\bang{\star})}}}{\outwc{\chanone}{\bang{\varone}}})}$:
once in the scope of the $!$ operator, once outside any $!$.  When a value is
substituted for $\varone$ (and thus duplicated) one of the two copies interacts 
with the other, being copied again. It is this cyclic phenomenon (called \emph{modal
impredicativity} in~\cite{Dallago09}) that is responsible for nontermination.

The Linear Logic community has studied in depth the impact of unbalanced and multiple
boxes on the complexity
of computation, 
and singled out several (different) sufficient conditions for ensuring not only termination, but
termination with prescribed bounds. We will adopt here the conditions arising from
Lafont's analysis (and formalized in Soft Linear Logic, \SLL~\cite{Lafont:TCS01}), leaving
to further work the usage of other criteria. We thus introduce the calculus \shopi\ of
\emph{soft processes}, for which we will prove termination in polynomial time. 
In our view, this is the main contribution of the paper. 

Soft processes share the same grammar and operational semantics than linear processes (Section~\ref{sect:lhopi}), but
are subjected to stronger constraints, expressed by the well-formation rules of Figure~\ref{fig:shopiwff}.
\begin{figure}
\begin{center}
\fbox{\begin{minipage}[c]{.97\textwidth}
\vspace{12pt}
$$
\begin{array}{ccccc}
\infer
{\wfjsp{\dies{\conone}}{\emproc}}
{}
&
\hspace{10pt}
&
\infer
{\wfjsp{\conone,\contwo,\dies{\conthree}}{\para{\procone}{\proctwo}}}
{\wfjsp{\conone,\dies{\conthree}}{\procone} & \wfjsp{\contwo,\dies{\conthree}}{\proctwo}}
&
\hspace{10pt}
&
\infer
{\wfjsp{\conone}{\inp{\chanone}{\varone}{\procone}}}
{\wfjsp{\conone,\varone}{\procone}}
\end{array}
$$
\vspace{2pt}
$$
\begin{array}{ccccc}
\infer
{\wfjsp{\conone}{\einp{\chanone}{\varone}{\procone}}}
{\wfjsp{\conone,\bang{\varone}}{\procone}}
&
\hspace{10pt}
&
\infer
{\wfjsp{\conone}{\einp{\chanone}{\varone}{\procone}}}
{\wfjsp{\conone,\dies{\varone}}{\procone}}
&
\hspace{10pt}
&
\infer
{\wfjsp{\conone,\contwo,\dies{\conthree}}{\out{\chanone}{\valueone}{\procone}}}
{\wfjsv{\conone,\dies{\conthree}}{\valueone} & \wfjsp{\contwo,\dies{\conthree}}{\procone}}
\end{array}
$$
\vspace{2pt}
$$
\begin{array}{ccccc}
\infer
{\wfjsp{\conone}{\rest{\chanone}{\procone}}}
{\wfjsp{\conone}{\procone}}
&
\hspace{10pt}
&
\infer
{\wfjsp{\conone,\contwo,\dies{\conthree}}{\app{\valueone}{\valuetwo}}}
{\wfjsv{\conone,\dies{\conthree}}{\valueone} & \wfjsv{\contwo,\dies{\conthree}}{\valuetwo}}
&
\hspace{10pt}
&
\infer
{\wfjsv{\dies{\conone}}{\star}}
{}
\end{array}
$$
\vspace{2pt}
$$
\begin{array}{ccccc}
\infer
{\wfjsv{\dies{\conone},\varone}{\varone}}
{}
&
\hspace{10pt}
&
\infer
{\wfjsv{\dies{\conone},\dies{\varone}}{\varone}}
{}
&
\hspace{10pt}
&
\infer
{\wfjsv{\conone}{\abstr{\varone}{\procone}}}
{\wfjsp{\conone,\varone}{\procone}}
\end{array}
$$
\vspace{2pt}
$$
\begin{array}{ccccc}
\infer
{\wfjsv{\conone}{\eabstr{\varone}{\procone}}}
{\wfjsp{\conone,\dies{\varone}}{\procone}}
&
\hspace{10pt}
&
\infer
{\wfjsv{\conone}{\eabstr{\varone}{\procone}}}
{\wfjsp{\conone,\bang{\varone}}{\procone}}
&
\hspace{10pt}
&
\infer
{\wfjsv{\bang{\conone},\dies{\contwo}}{\bang{\valueone}}}
{\wfjsv{\conone}{\valueone}}
\end{array}
$$
\vspace{7pt}
\end{minipage}}
\end{center}
\caption{Processes and values in \shopi.}\label{fig:shopiwff}
\end{figure}
A context $\conone$ can now contain a variable $\varone$ in at most
one of \emph{three} different forms: $\varone$, $\bang{\varone}$,
or $\dies{\varone}$. The implicit contraction  (or weakening)  happens
on $\dies{\relax}$-marked variables, but none of them may ever appear inside a !-box.
In the last rule it is implicitly assumed that the context $\conone$ in the premise
is composed only of linear variables, if any (otherwise the context $\bang{\conone}$
of the conclusion would be ill-formed).
Indeed, 
the rules amount to say that, if $\wfjsp{\conone}{\procone}$ (and similarly for values),
then: (i) any linear variable $\varone$ in $\conone$
occurs exactly once in $\procone$, and at depth $0$ (this is as in \lhopi);
(ii) any nonlinear variable
$\bang{\varone}$ occurs exactly once in $\procone$, and at depth $1$;
(iii) any nonlinear
variable $\dies{\varone}$ may occur any number of times in $\procone$, all of its occurrences must be at level $0$. As a result, any bound variable appears in the scope of the binder always at a same level.
As in \lhopi, 
well-formed processes are closed by reduction:
\begin{proposition}
If $\wfjsp{\;}{\procone}$ and $\osl{\procone}{\proctwo}$, then
$\wfjsp{\;}{\proctwo}$.
\end{proposition}

The nonterminating process $\lprocomega$ which started this section is \emph{not} a soft process, 
because the bound variable $x$ appears twice, once at depth $0$ and once depth $1$.  
And this is good news: we would like \shopi\ to be a calculus of terminating processes, at least! 
But this has some drawbacks: also
$\lserver$ is not a soft process. Indeed, \shopi\ is not able to discriminate between
$\lserver$ and $\lprocomega$, which share a very similar structure. We will come back to this
after we proved our main result on the polynomial bound on reduction sequences for soft processes.

\subsection{Feasible Termination}\label{sect:shopifr}
This section is devoted to the proof of feasible termination for soft processes. We prove that the length of
any reduction sequence from a soft process $\procone$ is bounded by a polynomial on the size of $\procone$.
Moreover, the size of any process along the reduction is itself polynomially bounded.

The proof proceeds similarly to the one for \SLL\ proof-nets by Lafont~\cite{Lafont:TCS01}.
The idea is relatively simple: a weight is assigned to every process and is proved to decrease at any normalization
step. The weight of a process can be proved to be an \emph{upper bound} on the size of the process. Finally, a polynomial
bound on the weight of a process holds. Altogether, this implies feasible termination.

Before embarking on the proofs, we need some preliminary definitions. First of all,
the \emph{size} of a process $\procone$ (respectively, a value $\valueone$) is defined simply as the number of symbols in it
and is denoted as $\size{\procone}$ (respectively, $\size{\valueone}$)
Another crucial attribute of processes and values is their \emph{box depth}, namely the maximum nesting of $!$ operators
inside them; for a process $\procone$ and a value $\valueone$, it is denoted either as $\bd{\procone}$ or as $\bd{\valueone}$.
The \emph{duplicability factor} $\df{\procone}$ of a process $\procone$ is the maximum number of free occurrences of a variable
$\varone$ for every binder in $\procone$; similarly for values. The precise definition follows, where $\nfo{\varone}{\procone}$ denotes the number of free occurrences on $\varone$ in
$\procone$.
\begin{align*}
\df{\star}=\df{\varone}=\df{\emproc}&=1; & 
\df{\abstr{\varone}{\procone}}=\df{\abstr{\bang{\varone}}{\procone}}&=\max\{\df{\procone},\nfo{\varone}{\procone}\};\\
\df{\bang{\valueone}}&=\df{\valueone}; &
\df{\para{\procone}{\proctwo}}&=\max\{\df{\procone},\df{\proctwo}\};\\
\df{\inp{\chanone}{\varone}{\procone}}=\df{\inp{\chanone}{\bang{\varone}}{\procone}}&=\max\{\df{\procone},\nfo{\varone}{\procone}\}; &
\df{\out{\chanone}{\valueone}{\procone}}&=\max\{\df{\valueone},\df{\procone}\};\\
\df{\rest{\chanone}{\procone}}&=\df{\procone}; &
\df{\app{\valueone}{\valuetwo}}&=\max\{\df{\valueone},\df{\valuetwo}\}.
\end{align*}
Finally, we can define the weight of processes and values. A notion of weight parametrized on a natural number $n$
can be given as follows, by induction on the structure of processes and values:
\begin{align*}
\weip{\star}{n}=\weip{\varone}{n}=\weip{\emproc}{n}&=1; &
\weip{\abstr{\varone}{\procone}}{n}=\weip{\abstr{\bang{\varone}}{\procone}}{n}&=\weip{\procone}{n};\\
\weip{\bang{\valueone}}{n}&=n\cdot\weip{\valueone}{n}+1; &
\weip{\para{\procone}{\proctwo}}{n}&=\weip{\procone}{n}+\weip{\proctwo}{n}+1;\\
\weip{\inp{\chanone}{\varone}{\procone}}{n}=\weip{\inp{\chanone}{\bang{\varone}}{\procone}}{n}&=\weip{\procone}{n}+1; &
\weip{\out{\chanone}{\valueone}{\procone}}{n}&=\weip{\valueone}{n}+\weip{\procone}{n};\\
\weip{\rest{\chanone}{\procone}}{n}&=\weip{\procone}{n}; &
\weip{\app{\valueone}{\valuetwo}}{n}&=\weip{\valueone}{n}+\weip{\valuetwo}{n}+1.
\end{align*}
Now, the \emph{weight} $\wei{\procone}$ of a process $\procone$ is $\weip{\procone}{\df{\procone}}$. Similarly for values.

The first auxiliary result is about structural congruence. As one would expect, two structurally
congruent terms have identical size, box depth, duplicability factor and weight:
\begin{proposition}\label{prop:shopisc}
if $\congr{\procone}{\proctwo}$, then
$\size{\procone}=\size{\proctwo}$, $\bd{\procone}=\bd{\proctwo}$, $\df{\procone}=\df{\proctwo}$. Moreover,
for every $n$, $\weip{\procone}{n}=\weip{\proctwo}{n}$.
\end{proposition}
Observe that Proposition~\ref{prop:shopisc} would not hold in presence of structural congruence rules like
$\para{\procone}{\emproc}\congrs\procone$ and $\rest{\chanone}{\emproc}\congrs\chanone$.

How does $\df{\procone}$ evolve during reduction? Actually, it cannot grow:
\begin{lemma}\label{lemma:shopidfni}
If $\wfjsp{\;}{\proctwo}$ and $\osl{\proctwo}{\procone}$, then
$\df{\proctwo}\geq\df{\procone}$.
\end{lemma}
\begin{proof}
As an auxiliary lemma, we can prove that
whenever $\wfjsp{\conone}{\procone}$ and $\wfjsv{\emcon}{\valueone},\wfjsv{\contwo}{\valuetwo}$,
both $\df{\subst{\procone}{\varone}{\valueone}}\leq\max\{\df{\procone},\df{\valueone}\}$ and 
$\df{\subst{\valuetwo}{\varone}{\valueone}}\leq\max\{\df{\valuetwo},\df{\valueone}\}$. This
is an easy induction on derivations
for $\wfjsp{\conone}{\procone}$ and $\wfjsv{\contwo}{\valuetwo}$. The thesis follows.
\end{proof}

The weight of a process is an upper bound to the size of the process itself.
This means that bounding the weight of a process implies bounding its
size. Moreover, the weight of a process strictly decreases at any reduction step.
\begin{lemma}\label{lemma:shopiweiub}
For every $\procone$, $\wei{\procone}\geq\size{\procone}$. 
\end{lemma}
\begin{proof}
By induction on $\procone$, strengthening the induction hypothesis
with a similar statement for values. In the induction, observe that
$\df{\procone},\df{\valueone}\geq 1$ for every process $\procone$
and value $\valueone$.
\end{proof}

\begin{proposition}\label{prop:shopiweidec}
If $\wfjsp{\;}{\proctwo}$ and $\osl{\proctwo}{\procone}$, then
$\wei{\proctwo}>\wei{\procone}$.
\end{proposition}
\begin{proof}
As an auxiliary result, we need to prove the following (slightly modifications of) substitution lemmas
(let $\wfjsv{\emcon}{\valueone}$ and $n\geq m\geq 1$):
\begin{varitemize}
\item
  If $\tdone:\wfjsp{\conone,\varone}{\procthree}$, then
  $\weip{\subst{\procthree}{\varone}{\valueone}}{m}\leq \weip{\procthree}{n}+\weip{\valueone}{n}$;
\item
  If $\tdone:\wfjsv{\conone,\varone}{\valuetwo}$, then
  $\weip{\subst{\valuetwo}{\varone}{\valueone}}{m}\leq \weip{\valuetwo}{n}+\weip{\valueone}{n}$;
\item
  If $\tdone:\wfjsp{\conone,\dies{\varone}}{\procthree}$, then
  $\weip{\subst{\procthree}{\varone}{\valueone}}{m}\leq \weip{\procthree}{n}+\nfo{\varone}{\procthree}\cdot\weip{\valueone}{n}$;
\item
  If $\tdone:\wfjsv{\conone,\dies{\varone}}{\valuetwo}$, then
  $\weip{\subst{\valuetwo}{\varone}{\valueone}}{m}\leq \weip{\valuetwo}{n}+\nfo{\varone}{\valuetwo}\cdot\weip{\valueone}{n}$;
\item
  If $\tdone:\wfjsp{\conone,\bang{\varone}}{\procthree}$, then
  $\weip{\subst{\procthree}{\varone}{\valueone}}{m}\leq \weip{\procthree}{n}+n\cdot\weip{\valueone}{n}$;
\item
  If $\tdone:\wfjsv{\conone,\bang{\varone}}{\valuetwo}$, then
  $\weip{\subst{\valuetwo}{\varone}{\valueone}}{m}\leq \weip{\valuetwo}{n}+n\cdot\weip{\valueone}{n}$;
\end{varitemize}
This is an induction on $\tdone$. An inductive case:
\begin{varitemize}
\item
  If $\tdone$ is:
  $$
  \infer
      {\wfjsv{\bang{\conone},\bang{\varone},\dies{\contwo}}{\bang{\valuethree}}}
      {\wfjsv{\conone,\varone}{\valuethree}}
  $$
  then $\valuetwo=\bang{\valuethree}$ and $\subst{(\bang{\valuethree})}{\varone}{\valueone}$ is simply
  $\bang{(\subst{\valuethree}{\varone}{\valueone})}$. As a consequence:
  \begin{align*}
  \weip{\subst{\valuetwo}{\varone}{\valueone}}{m}&=m\cdot\weip{\subst{\valuethree}{\varone}{\valueone}}{m}+1
      \leq n\cdot(\weip{\valuethree}{n}+\weip{\valueone}{n})+1=n\cdot\weip{\valuethree}{n}+n\cdot\weip{\valueone}{n}+1\\
    &=\weip{\bang{\valuethree}}{n}+n\cdot\weip{\valueone}{n}=\weip{\valuetwo}{n}+n\cdot\weip{\valueone}{n}.
  \end{align*}
\end{varitemize}
With the above observations in hand, we can easily prove the thesis by induction
on any derivation $\tdtwo$ of $\os{\procone}{\proctwo}$:
\begin{varitemize}
\item
  Suppose $\tdtwo$ is
  $$
  \infer
  {\osl{\para{\out{\chanone}{\valueone}{\procthree}}{\inp{\chanone}{\varone}{\procfour}}}
       {\para{\procthree}{\subst{\procfour}{\varone}{\valueone}}}}
  {}
  $$
  From $\wfjsp{\emcon}{\para{\out{\chanone}{\valueone}{\procthree}}{\inp{\chanone}{\varone}{\procfour}}}$, it
  follows that $\wfjsp{\emcon}{\procthree}$, $\wfjsv{\emcon}{\valueone}$
  and $\wfjsp{\varone}{\procfour}$. As a consequence, since $\df{\proctwo}\leq\df{\procone}$, 
  \begin{align*}
    \wei{\procone}&=\wei{\para{\out{\chanone}{\valueone}{\procthree}}{\inp{\chanone}{\varone}{\procfour}}}
      =\weip{\valueone}{\df{\procone}}+\weip{\procthree}{\df{\procone}}+\weip{\procfour}{\df{\procone}}+2\\
      &\geq\weip{\subst{\procfour}{\varone}{\valueone}}{\df{\proctwo}}+\weip{\procthree}{\df{\proctwo}}+2
      >\weip{\subst{\procfour}{\varone}{\valueone}}{\df{\proctwo}}+\weip{\procthree}{\df{\proctwo}}+1
      =\weip{\para{\subst{\procfour}{\varone}{\valueone}}{\procthree}}{\df{\proctwo}}.
  \end{align*}
\item
  Suppose $\tdtwo$ is
  $$
  \infer
  {\osl{\para{\out{\chanone}{\bang{\valueone}}{\procthree}}{\einp{\chanone}{\varone}{\procfour}}}
       {\para{\procthree}{\subst{\procfour}{\varone}{\valueone}}}}
  {}
  $$
  From $\wfjsp{\emcon}{\para{\out{\chanone}{\valueone}{\procthree}}{\inp{\chanone}{\varone}{\procfour}}}$, it
  follows that $\wfjsp{\emcon}{\procthree}$, $\wfjsv{\emcon}{\valueone}$
  and either $\wfjsp{\bang{\varone}}{\procfour}$ or $\wfjsp{\dies{\varone}}{\procfour}$. 
  In the first case:
  \begin{align*}
    \wei{\procone}&=\wei{\para{\out{\chanone}{\bang{\valueone}}{\procthree}}{\inp{\chanone}{\varone}{\procfour}}}
      =\weip{\bang{\valueone}}{\df{\procone}}+\weip{\procthree}{\df{\procone}}+\weip{\procfour}{\df{\procone}}+2\\
      &=\df{\procone}\cdot\weip{\valueone}{\df{\procone}}+\weip{\procthree}{\df{\procone}}+\weip{\procfour}{\df{\procone}}+3
      \geq\weip{\subst{\procfour}{\varone}{\valueone}}{\df{\proctwo}}+\weip{\procthree}{\df{\proctwo}}+3\\
      &>\weip{\subst{\procfour}{\varone}{\valueone}}{\df{\proctwo}}+\weip{\procthree}{\df{\proctwo}}+1
      =\weip{\para{\subst{\procfour}{\varone}{\valueone}}{\procthree}}{\df{\proctwo}}.
  \end{align*}
  In the second case:
  \begin{align*}
    \wei{\procone}&=\wei{\para{\out{\chanone}{\bang{\valueone}}{\procthree}}{\inp{\chanone}{\varone}{\procfour}}}
      =\weip{\bang{\valueone}}{\df{\procone}}+\weip{\procthree}{\df{\procone}}+\weip{\procfour}{\df{\procone}}+2\\
      &=\df{\procone}\cdot\weip{\valueone}{\df{\procone}}+\weip{\procthree}{\df{\procone}}+\weip{\procfour}{\df{\procone}}+3\\
      &\geq\nfo{\varone}{\procfour}\cdot\weip{\valueone}{\df{\procone}}+\weip{\procthree}{\df{\procone}}+\weip{\procfour}{\df{\procone}}+3\\
      &\geq\weip{\subst{\procfour}{\varone}{\valueone}}{\df{\proctwo}}+\weip{\procthree}{\df{\proctwo}}+3
      >\weip{\subst{\procfour}{\varone}{\valueone}}{\df{\proctwo}}+\weip{\procthree}{\df{\proctwo}}+1
      =\weip{\para{\subst{\procfour}{\varone}{\valueone}}{\procthree}}{\df{\proctwo}}.
  \end{align*}
\item
  Suppose $\tdtwo$ is
  $$
  \infer
  {\osl{\para{\procthree}{\procfive}}{\para{\procfour}{\procfive}}}
  {\tdthree:\osl{\procthree}{\procfour}}
  $$
  From  $\wfjsp{\emcon}{\para{\procthree}{\procfive}}$, it follows
  that $\wfjsp{\emcon}{\procthree}$ and $\wfjsp{\emcon}{\procfive}$.
  By induction hypothesis on $\tdthree$, this yields 
  $\wei{\procthree}>\wei{\procfour}$,
  and in turn $\wei{\procthree}=\wei{\procthree}+\wei{\procfive}+1>\wei{\procfour}+\wei{\procfive}+1=\wei{\procfour}$.
\end{varitemize}
This concludes the proof.
\end{proof}

Lemma~\ref{lemma:shopiweiub} and Proposition~\ref{prop:shopiweidec} together imply that the
weight is an upper bound to both the number of reduction steps a process can perform and
the size of any reduct. So, the only missing tale is bounding the weight itself:
\begin{proposition}
For every process $\procone$, $\wei{\procone}\leq\size{\procone}^{\bd{\procone}+1}$. 
\end{proposition}
\begin{proof}
By induction on $\procone$, enriching the thesis with an analogous statement for values:
$\wei{\valueone}\leq\size{\valueone}^{\bd{\valueone}+1}$.
\end{proof}

Putting all the ingredients together, we reach our soundness result with respect polynomial time:
\begin{theorem}\label{theo:soundshopi}
There is a family of polynomials $\{p_n\}_{n}$ such that for every process $\procone$ and for every
$m$, if $\oslp{\procone}{m}{\proctwo}$, then $m,\size{\proctwo}\leq p_{\bd{\procone}}(\size{\procone})$.
\end{theorem}
The polynomials in Theorem \ref{theo:soundshopi} depend on terms, so the bound on the number
of internal actions is not polynomial, strictly speaking. Please observe, however, that all processes
with the same box depth $b$ are governed by the same polynomial $p_b$, similarly to what happens in 
Soft Linear Logic.

\subsection{Completeness?}
Soundness of a formal system with respect to some semantic criterion is useless unless
one shows that the system is also \emph{expressive enough}. 
In implicit computational complexity, programming languages 
are usually proved both sound and \emph{extensionally complete} with respect to a complexity class. 
Not only any program can be normalized in bounded
time, but every function in the class can be computed by a program in the system.
Preliminary to any completeness result for \shopi, however, would be the definition of
what a complexity class for processes should be (as opposed to the well known definition for
functions or problems). This is an elusive---and very interesting---problem that we cannot
tackle in this preliminary work and that we leave for future work. 

Certainly the expressiveness of \shopi\ is weak if we take into
account the visible actions  of the processes (i.e., their interactions
with the environment). This is due to the limited possibilities of
copying, and hence also of writing recursive process
 behaviours. Indeed, one cannot consider \shopi, on its own, as a general-purpose calculus
 for concurrency. However, we believe that the study  of \shopi,
 or similar languages, could  be fruitful in establishing bounds on
 the internal behaviour of parts, or components, of a concurrent
 systems; for instance, on the time and space that a process may take
 to answer a query from another process (in this case the \shopi\ techniques
would be applied  to the parts of the syntax of the process that
describe its internal computation after the query).
Next section considers a possible direction of development of
\shopi, allowing more freedom on the external actions of the processes.

We are convinced, on the other hand, that a minimal completeness result can be given,
namely the possibility of representing all polynomial time \emph{functions} (or
problems) in \shopi. Possibly, this could be done by encoding Soft Linear Logic into
\shopi\ through a continuation-passing style translation. We leave this 
to future work.

\section{An Extension to \shopi:   Spawning}\label{sect:eshopi}

In this section we propose an extension of \shopi\ that allows us to
accept processes such as $\lserver$, capable of performing 
infinitely many interactions with their external environment while
maintaining  polynomial bounds on the number of internal steps they
can make between
any two external actions.

The reason why  $\lserver$  is \emph{not} a
\shopi\ process has to do with the bound variable $\varone$ in the sub-process $\lcompon$:
$$
\lcompon=\eabstr{\varthree}{( 
\einp{\chanone}{\varone}{(\para{\einp{\chantwo}{\vartwo}
    {\out{\chanthree}{\bang{\vartwo}} {\app \varone {(\bang{\star})}}}} 
\outwc{\chanone}{\bang{\varone}})})},
$$
The variable appears twice in the body
$(\para{\einp{\chantwo}{\vartwo}{\out{\chanthree}{\bang{\vartwo}}{\app \varone {(\bang{\star})}}}} 
\outwc{\chanone}{\bang{\varone}})$, at two different !-depths.  This pattern is not 
permitted in \shopi, because otherwise also the nonterminating process
 $\lprocomega$ would be in the calculus. 
There is however a major difference between 
$\lprocomega$ and $\lserver$:
in $\lcompon$, one of the two occurrences
of $\varone$ (the one at depth $0$) is part of the continuation of an input
on $\chantwo$; moreover, such channel  $\chantwo$ is only used by
$\lserver$ in
input~--- $\lserver$ does not own the output capability.
This implies that  whatever process  will  substitute that occurrence of $\varone$,
it will be able to interact with the environment only \emph{after} an input on $\chantwo$
is performed. So, its ``computational weight'' does not affect the number of reduction steps
made by the process \emph{before} such an input occurs. This phenomenon, which does not occur 
in $\lprocomega$, can be seen as a form of process spawning: $\lcompon$ can be copied
an unbounded number of times, but  the rhythm of the copying is dictated
by the  input actions at $\chantwo$. 

Consider a subset $\inpchans$ of $\chans$ (where $\chans$ is the set of all channels which
can appear in processes). The process calculus \eshopi{\inpchans} is an extension of 
\shopi\ parametrized on $\inpchans$. What \eshopi{\inpchans} adds to
\shopi\ is precisely the possibility  of marking a subprocess as a component which
can be spawned. This is accomplished with a new operator $\Box$. 
Channels in $\inpchans$ are called \emph{input channels}, because outputs
are forbidden on them. The syntax of processes and values is enriched as follows:
\begin{align*}
\procone&::=\ldots\mid\sinp{\chanone}{\varone}{\procone};\\
\valueone&::=\ldots\mid\sabstr{\varone}{\procone}\mid\sang{\valueone};
\end{align*}
but outputs can only be performed on channels not in $\inpchans$.
The term $\sang{\valueone}$ is a value (i.e., a parametrized process) which can be
spawned. Spawning itself is performed by passing a process $\sang{\valueone}$
to either an abstraction $\sabstr{\varone}{\procone}$ or an input
$\sinp{\chanone}{\varone}{\procone}$. 
In both cases, exactly one occurrence of $\varone$ in $\procone$ is the scope
of a $\Box$ operator, and only one of the following 
two conditions holds:
\begin{varenumerate}
\item
  The occurrence of $\varone$ in the scope of a $\Box$ operator is part of the continuation of an input channel $\chanone$, and all other occurrences of $\varone$ in $\procone$ are at depth $0$.
\item
  There are no other occurrences of $\varone$ in $\procone$.
\end{varenumerate}
The foregoing constraints are enforced by the well-formation rules in
Figure~\ref{fig:eshopiwff}.
\begin{figure}
\begin{center}
\fbox{\begin{minipage}[c]{.97\textwidth}
\vspace{12pt}
$$
\begin{array}{ccccc}
\infer
{\wfjep{\dies{\conone}}{\emproc}}
{}
&
\hspace{10pt}
&
\infer
{\wfjep{\conone,\contwo,\dies{\conthree},\dang{\confour}}{\para{\procone}{\proctwo}}}
{\wfjep{\conone,\dies{\conthree},\dang{\confour}}{\procone} & \wfjep{\contwo,\dies{\conthree},\dies{\confour}}{\proctwo}}
&
\hspace{10pt}
&
\infer
{\wfjep{\conone}{\inp{\chanone}{\varone}{\procone}}}
{\wfjep{\conone,\varone}{\procone}}
\end{array}
$$
\vspace{2pt}
$$
\begin{array}{ccccccc}
\infer
{\wfjep{\conone}{\einp{\chanone}{\varone}{\procone}}}
{\wfjep{\conone,\bang{\varone}}{\procone}}
&
\hspace{10pt}
&
\infer
{\wfjep{\conone}{\einp{\chanone}{\varone}{\procone}}}
{\wfjep{\conone,\dies{\varone}}{\procone}}
&
\hspace{10pt}
&
\infer
{\wfjep{\conone}{\sinp{\chanone}{\varone}{\procone}}}
{\wfjep{\conone,\sang{\varone}}{\procone}}
&
\hspace{10pt}
&
\infer
{\wfjep{\conone}{\sinp{\chanone}{\varone}{\procone}}}
{\wfjep{\conone,\dang{\varone}}{\procone}}
\end{array}
$$
\vspace{2pt}
$$
\begin{array}{ccc}
\infer
{\wfjep{\conone,\dang{\contwo}}{\inp{\chanone}{\varone}{\procone}}}
{\wfjep{\conone,\sang{\contwo},\varone}{\procone} & \chanone\in\inpchans}
&
\hspace{10pt}
&
\infer
{\wfjep{\conone,\dang{\contwo}}{\einp{\chanone}{\varone}{\procone}}}
{\wfjep{\conone,\sang{\contwo},\bang{\varone}}{\procone} & \chanone\in\inpchans}
\end{array}
$$
\vspace{2pt}
$$
\begin{array}{ccccc}
\infer
{\wfjep{\conone,\dang{\contwo}}{\einp{\chanone}{\varone}{\procone}}}
{\wfjep{\conone,\sang{\contwo},\dies{\varone}}{\procone} & \chanone\in\inpchans}
&
\hspace{10pt}
&
\infer
{\wfjep{\conone,\dang{\contwo}}{\sinp{\chanone}{\varone}{\procone}}}
{\wfjep{\conone,\sang{\contwo},\sang{\varone}}{\procone} & \chanone\in\inpchans}
&
\hspace{10pt}
&
\infer
{\wfjep{\conone,\dang{\contwo}}{\sinp{\chanone}{\varone}{\procone}}}
{\wfjep{\conone,\sang{\contwo},\dang{\varone}}{\procone} & \chanone\in\inpchans}
\end{array}
$$
\vspace{2pt}
$$
\begin{array}{ccc}
\infer
{\wfjep{\conone,\contwo,\dies{\conthree},\dang{\confour}}{\out{\chanone}{\valueone}{\procone}}}
{\wfjev{\conone,\dies{\conthree},\dang{\confour}}{\valueone} & \wfjep{\contwo,\dies{\conthree},\dies{\confour}}{\procone}}
&
\hspace{10pt}
&
\infer
{\wfjep{\conone,\contwo,\dies{\conthree},\dang{\confour}}{\out{\chanone}{\valueone}{\procone}}}
{\wfjev{\conone,\dies{\conthree},\dies{\confour}}{\valueone} & \wfjep{\contwo,\dies{\conthree},\dang{\confour}}{\procone}}
\end{array}
$$
\vspace{2pt}
$$
\begin{array}{ccccc}
\infer
{\wfjep{\conone}{\rest{\chanone}{\procone}}}
{\wfjep{\conone}{\procone}}
&
\hspace{5pt}
&
\infer
{\wfjep{\conone,\contwo,\dies{\conthree},\dang{\confour}}{\app{\valueone}{\valuetwo}}}
{\wfjev{\conone,\dies{\conthree},\dang{\confour}}{\valueone} & \wfjev{\contwo,\dies{\conthree},\dies{\confour}}{\valuetwo}}
&
\hspace{5pt}
&
\infer
{\wfjep{\conone,\contwo,\dies{\conthree},\dang{\confour}}{\app{\valueone}{\valuetwo}}}
{\wfjev{\conone,\dies{\conthree},\dies{\confour}}{\valueone} & \wfjev{\contwo,\dies{\conthree},\dang{\confour}}{\valuetwo}}
\end{array}
$$
\vspace{2pt}
$$
\begin{array}{ccccccc}
\infer
{\wfjev{\dies{\conone}}{\star}}
{}
&
\hspace{10pt}
&
\infer
{\wfjev{\dies{\conone},\varone}{\varone}}
{}
&
\hspace{10pt}
&
\infer
{\wfjev{\dies{\conone},\dies{\varone}}{\varone}}
{}
&
\hspace{10pt}
&
\infer
{\wfjev{\conone}{\abstr{\varone}{\procone}}}
{\wfjep{\conone,\varone}{\procone}}
\end{array}
$$
\vspace{2pt}
$$
\begin{array}{ccccc}
\infer
{\wfjev{\conone}{\eabstr{\varone}{\procone}}}
{\wfjep{\conone,\dies{\varone}}{\procone}}
&
\hspace{10pt}
&
\infer
{\wfjev{\conone}{\eabstr{\varone}{\procone}}}
{\wfjep{\conone,\bang{\varone}}{\procone}}
&
\hspace{10pt}
&
\infer
{\wfjev{\conone}{\sabstr{\varone}{\procone}}}
{\wfjep{\conone,\sang{\varone}}{\procone}}
\end{array}
$$
\vspace{2pt}
$$
\begin{array}{ccccc}
\infer
{\wfjev{\conone}{\sabstr{\varone}{\procone}}}
{\wfjep{\conone,\dang{\varone}}{\procone}}
&
\hspace{10pt}
&
\infer
{\wfjev{\bang{\conone},\dies{\contwo}}{\bang{\valueone}}}
{\wfjev{\conone}{\valueone}}
&
\hspace{10pt}
&
\infer
{\wfjev{\sang{\conone},\dies{\contwo}}{\sang{\valueone}}}
{\wfjev{\conone}{\valueone}}
\end{array}
$$
\vspace{7pt}
\end{minipage}}
\end{center}
\caption{Processes and values in \eshopi{\inpchans}.}\label{fig:eshopiwff}
\end{figure}
The well-formation rules of \eshopi{\inpchans} are considerably more complex than the ones of \shopi. 
Judgements have the form $\wfjep{\conone}{\procone}$ or
$\wfjev{\conone}{\valueone}$, where a variable $\varone$ can occur in $\conone$
in one of five different forms:
\begin{varitemize}
\item
  As either $\varone$, $\bang{\varone}$ or $\dies{\varone}$: here the meaning
  is exactly the one from \shopi\ (see Section~\ref{sect:shopi}).
\item
  As $\sang{\varone}$: the variable $\varone$ then appears exactly once in $\procone$, in
  the scope of a spawning operator $\Box$.
\item
  As $\dang{\varone}$: $\varone$ occurs at least once in $\procone$, once in the
  scope of a $\Box$ operator (itself part of the continuation for an input channel),
  and possibly many times at depth $0$.
\end{varitemize}
A variable marked as $\dang{\varone}$ can ``absorb'' the same
variable declared as $\dies{\varone}$ in binary well-formation rules (i.e.
the ones for applications, outputs, etc.). Note 
the special well-formation rules  that are only applicable with an
 input channel: in that case a portion of the context 
$\sang{\contwo}$ becomes $\dang{\contwo}$.

The operational semantics is obtained adding to Figure~\ref{fig:oslhopi} the following
two rules:
$$
\begin{array}{ccc}
\infer
{\osl{\para{\out{\chanone}{\sang{\procthree}}{\procone}}{\sinp{\chanone}{\varone}{\proctwo}}}
{\para{\procone}{\subst{\proctwo}{\varone}{\procthree}}}}
{}
&
\hspace{10pt}
&
\infer
{\osl{\app{(\sabstr{\varone}{\procone})}{\sang{\proctwo}}}{\subst{\procone}{\varone}{\proctwo}}}
{}
\end{array}
$$
As expected,
\begin{lemma}[Subject Reduction]
If $\wfjep{\;}{\procone}$ and $\osl{\procone}{\proctwo}$, then $\wfjep{\;}{\proctwo}$.
\end{lemma}

The process $\lserver$ is a \eshopi{\inpchans}\ process once $\lcompon$ is
considered as a spawned process and $\chantwo\in\inpchans$: define
\begin{align*}
\elserver&=\rest{\chanone}{(\para{\app \elcompon{(\bang{\star})}}{\outwc{\chanone}{\sang{\elcompon}}})};\\
\elcompon&=\eabstr{\varthree}{( 
\sinp{\chanone}{\varone}{(\para{\einp{\chantwo}{\vartwo}
    {\out{\chanthree}{\bang{\vartwo}} {\outwc{\chanone}{\sang{\varone}}}}}{\app{\varone}{(\bang{\star})}})})}.
\end{align*}
and consider the following derivations:
$$
\infer
{\wfjep{\emcon}{\rest{\chanone}{(\para{\app \elcompon{(\bang{\star})}}{\outwc{\chanone}{\sang{\elcompon}}})}}}
{
  \infer
  {\wfjep{\emcon}\para{\app \elcompon{(\bang{\star})}}{\outwc{\chanone}{\sang{\elcompon}}}}
  {
    \infer
    {\wfjep{\emcon}{\app{\elcompon}{(\bang{\star})}}}
    {
      \wfjev{\emcon}{\elcompon}
      &
      \infer
      {\wfjev{\emcon}{\bang{\star}}}
      {
        \wfjev{\emcon}{\star}
      }
    }
    &
    \infer
    {\wfjep{\emcon}{\outwc{\chanone}{\sang{\elcompon}}}}
    {
      \infer
      {\wfjev{\emcon}{\sang{\elcompon}}}
      {
        \wfjev{\emcon}{\elcompon}
      }
    }    
  }
}
$$
$$
\infer
{\wfjev{\emcon}{\eabstr{\varthree}{\sinp{\chanone}{\varone}{(\para{\einp{\chantwo}{\vartwo}
    {\out{\chanthree}{\bang{\vartwo}} {\outwc{\chanone}{\sang{\varone}}}}}{\app{\varone}{(\bang{\star})}})}}}}
{
  \infer
  {\wfjep{\dies{\varthree}}{\sinp{\chanone}{\varone}{(\para{\einp{\chantwo}{\vartwo}
    {\out{\chanthree}{\bang{\vartwo}} {\outwc{\chanone}{\sang{\varone}}}}}{\app{\varone}{(\bang{\star})}})}}}
  {
    \infer
    {
      \wfjep{\dies{\varthree},\dang{\varone}}
      {\para{\einp{\chantwo}{\vartwo}
        {\out{\chanthree}{\bang{\vartwo}} {\outwc{\chanone}{\sang{\varone}}}}}{\app{\varone}{(\bang{\star})}}}}
    {
      \infer
        {\wfjep{\dies{\varthree},\dang{\varone}}
           {\einp{\chantwo}{\vartwo}
             {\out{\chanthree}{\bang{\vartwo}} {\outwc{\chanone}{\sang{\varone}}}}}}
        {
          \infer
          {\wfjep{\dies{\varthree},\sang{\varone},\bang{\vartwo}}{\out{\chanthree}{\bang{\vartwo}} {\outwc{\chanone}{\sang{\varone}}}}}
          {
            \infer
            {\wfjep{\dies{\varthree},\sang{\varone}}{\outwc{\chanone}{\sang{\varone}}}}
            {
              \infer
              {\wfjev{\dies{\varthree},\sang{\varone}}{\sang{\varone}}}
              {
                \infer
                {\wfjev{\varone}{\varone}}
                {}
              }
            }
            &
            \infer
            {\wfjev{\bang{\vartwo}}{\bang{\vartwo}}}
            {
              \infer
              {\wfjev{\vartwo}{\vartwo}}
              {}
            }
          }
          &
          \chantwo\in\inpchans
        }
      &
      \infer
        {\wfjep{\dies{\varone}}{\app{\varone}{(\bang{\star})}}}
        {
          \infer
          {\wfjev{\dies{\varone}}{\varone}}
          {}
          &
          \infer
          {\wfjev{\emcon}{\bang{\star}}}
          {\wfjev{\emcon}{\star}}
        }
    }
  }
}
$$

The use in \eshopi{\inpchans} of a distinct set of input channels
may still be seen as rigid. For instance, it prevents from
accepting $\elserver$ in parallel with a client of the server itself
(because the client uses the request channel of the server in output);
similarly, it prevents from accepting reentrant servers (servers
that can invoke themselves). 
As pointed out earlier, we are mainly interested in techniques
capable of  ensuring polynomial bounds on \emph{components} of concurrent
systems (so for instance, bounds on the server, rather than on the composition
of the server and a client). In any case, 
this paper represents a preliminary
investigation, and further refinements or extensions of
\eshopi{\inpchans} may well be possible.
\subsection{Feasible Termination}
The proof of feasible termination for \eshopi{\inpchans} is similar in structure to the one
for \shopi\ (see Section~\ref{sect:shopifr}). However, some additional difficulties
due to the presence of spawning arise.

The auxiliary notions we needed in the proof of feasible termination for \shopi\ can
be easily extended to \eshopi{\inpchans} as follows:
The architecture of the soundness proof is similar to the one for linear processes.
The box depth, duplicability factor and weight of a process are defined as for soft processes, plus:

\vspace{-10pt}
{\scriptsize
\begin{align*}
\bd{\sabstr{\varone}{\procone}}&=\bd{\procone}; &
\df{\sabstr{\varone}{\procone}}&=\max\{\df{\procone},\nfo{\varone}{\procone}\}; &
\weip{\sabstr{\varone}{\procone}}{n}&=\weip{\procone}{n}; 
\\
\bd{\sang{\valueone}}&=\bd{\valueone}+1; &
\df{\sang{\valueone}}&=\df{\valueone}; &
\weip{\sang{\valueone}}{n}&=n\cdot\weip{\valueone}{n}+1; 
\\
\bd{\sinp{\chanone}{\varone}{\procone}}&=\bd{\procone}; &
\df{\sinp{\chanone}{\varone}{\procone}}&=\max\{\df{\procone},\nfo{\varone}{\procone}\}; &
\weip{\sinp{\chanone}{\varone}{\procone}}{n}&=\weip{\procone}{n}+1.
\end{align*}}
\vspace{-15pt}

\noindent
Informally, the spawning operator $\Box$ acts as $!$ in all the definitions above.
The weight $\wei{\procone}$, still defined as $\weip{\procone}{\df{\procone}}$ is 
again an upper bound to the size of $\procone$, but is not guaranteed to decrease
at any reduction step. In particular, spawning can make $\wei{\procone}$ bigger.
As a consequence, two new auxiliary notions are needed. The first one is
similar to the weight of processes and values, but is computed without taking
into account whatever happens after an input on a channel $\chanone\in\inpchans$. 
It is parametric on a natural number $n$ and is defined as follows:

\vspace{-10pt}
{\scriptsize\begin{align*}
\webip{\star}{n}=\webip{\varone}{n}=\webip{\emproc}{n}&=1 &
\webip{\abstr{\varone}{\procone}}{n}=\webip{\eabstr{\varone}{\procone}}{n}=\webip{\sabstr{\varone}{\procone}}{n}&=\webip{\procone}{n}\\
\webip{\bang{\valueone}}{n}=\webip{\sang{\valueone}}{n}&=n\cdot\webip{\valueone}{n}+1 &
\webip{\para{\procone}{\proctwo}}{n}&=\webip{\procone}{n}+\webip{\proctwo}{n}+1\\
\webip{\inp{\chanone}{\varone}{\procone}}{n}=\webip{\einp{\chanone}{\varone}{\procone}}{n}=\webip{\sinp{\chanone}{\varone}{\procone}}{n}&=
  \left\{
    \begin{array}{ll}
      0 & \mbox{if $\chanone\in\inpchans$}\\
      \webip{\procone}{n}+1 & \mbox{otherwise}\\
    \end{array}
  \right. &
\webip{\out{\chanone}{\valueone}{\procone}}{n}&=\webip{\valueone}{n}+\webip{\procone}{n} \\
\webip{\rest{\chanone}{\procone}}{n}&=\webip{\procone}{n} &
\webip{\app{\procone}{\proctwo}}{n}&=\webip{\procone}{n}+\webip{\proctwo}{n}+1
\end{align*}}
\vspace{-15pt}

\noindent
The \emph{weight before input} $\webi{\procone}$ of a process $\procone$ is simply 
$\webip{\procone}{\df{\procone}}$. As we will see, $\webi{\procone}$ \emph{is} guaranteed
to decrease at any reduction step, but this time it is not an upper bound to the size of
the underlying process. The second auxiliary notion captures the potential growth of
processes due to spawning and is again parametric on a natural number $n$:

\vspace{-10pt}
{\scriptsize\begin{align*}
\pgrp{\star}{n}=\pgrp{\varone}{n}=\pgrp{\emproc}{n}&=0 &
\pgrp{\abstr{\varone}{\procone}}{n}=\pgrp{\eabstr{\varone}{\procone}}{n}=\pgrp{\sabstr{\varone}{\procone}}{n}&=\pgrp{\procone}{n}\\
\pgrp{\bang{\valueone}}{n}&=n\cdot\pgrp{\valueone}{n} &
\pgrp{\sang{\valueone}}{n}&=n\cdot\pgrp{\valueone}{n}+n\cdot\weip{\valueone}{n}\\
\pgrp{\para{\procone}{\proctwo}}{n}&=\pgrp{\procone}{n}+\pgrp{\proctwo}{n} &
\pgrp{\inp{\chanone}{\varone}{\procone}}{n}=\pgrp{\einp{\chanone}{\varone}{\procone}}{n}=\pgrp{\sinp{\chanone}{\varone}{\procone}}{n}&=
  \left\{
    \begin{array}{ll}
      0 & \mbox{if $\chanone\in\inpchans$}\\
      \pgrp{\procone}{n} & \mbox{otherwise}\\
    \end{array}
  \right.\\
\pgrp{\out{\chanone}{\valueone}{\procone}}{n}&=\pgrp{\valueone}{n}+\pgrp{\procone}{n} &
\pgrp{\rest{\chanone}{\procone}}{n}&=\pgrp{\procone}{n}\\
\pgrp{\app{\valueone}{\valuetwo}}{n}&=\pgrp{\valueone}{n}+\pgrp{\valuetwo}{n}
\end{align*}}
\vspace{-10pt}

\noindent
Again, the \emph{potential growth} $\pgr{\procone}$ of a process $\procone$ is 
$\pgrp{\procone}{\df{\procone}}$.
Proposition~\ref{prop:shopisc}, Lemma~\ref{lemma:shopidfni} and Lemma~\ref{lemma:shopiweiub} from
Section~\ref{sect:shopifr} continue to hold for \eshopi{\inpchans}, and
their proofs remain essentially unchanged. Proposition~\ref{prop:shopiweidec} is true only
if the weight before input replaces the weight:
\begin{proposition}\label{prop:eshopiweidec}
If $\wfjsp{\emcon}{\proctwo}$ and $\osl{\proctwo}{\procone}$, then
$\webi{\proctwo}>\webi{\procone}$.
\end{proposition}
The potential growth of a process $\procone$ cannot increase during reduction. Moreover, 
the weight can increase, but at most by the decrease in the potential growth. Formally:
\begin{proposition}\label{prop:eshopiwebipgr}
If $\wfjsp{\emcon}{\proctwo}$ and $\osl{\proctwo}{\procone}$, then
$\pgr{\proctwo}\geq\pgr{\procone}$ and $\wei{\proctwo}+\pgr{\proctwo}\geq\wei{\procone}+\pgr{\procone}$. 
\end{proposition}
Polynomial bounds on all the attributes of processes we have defined can be proved: 
\begin{proposition}\label{prop:bounds}
For every process $\procone$, $\wei{\procone}\leq\size{\procone}^{\bd{\procone}+1}$,
$\webi{\procone}\leq\size{\procone}^{\bd{\procone}+1}$ and $\pgr{\procone}\leq\bd{\procone}\wei{\procone}$.
\end{proposition}
And, as for \shopi, we get a polynomial bound in the number of reduction steps from any process:
\begin{theorem}
There is a family of polynomials $\{p_n\}_{n\in\mathbb{N}}$ such that for every process $\procone$ and for every
$m$, if $\oslp{\procone}{m}{\proctwo}$, then $m,\size{\proctwo}\leq p_{\bd{\procone}}(\size{\procone})$.
\end{theorem}
Proofs for the results above have been elided, due to space constraints. Their structure, however, reflects
the corresponding proofs for \shopi\ (see Section~\ref{sect:shopifr}). As an example, proofs of propositions
\ref{prop:eshopiweidec} and \ref{prop:eshopiwebipgr} are both structured around appropriate substitution lemmas.
\section{Conclusions}
Goal of this preliminary essay was to verify whether we could apply to process
algebras the technologies for resource control that have been developed in
the so-called ``light logics'' and have been successfully applied so far
to paradigmatic functional programming. We deliberately adopted a minimalistic approach: 
applications between processes restricted to values,
the simplest available logic, a purely linear language (i.e., no weakening/erasing
on non marked formulas), no types, no search for maximal expressivity. 
In this way the result of the experiment would have had a clear single outcome.
We believe this outcome is a clear positive, and that this paper demonstrates it. 

Several issues must be investigated further, of course, so that this first experiment may
become a solid contribution. First,
one may wonder whether other complexity conscious fragments of linear logic can be used in place 
of \SLL\ as guideline for box control. \SLL\ is handy as a first try, because of
its simplicity, but we do believe that analogous results could be obtained starting
from Light Affine Logic, designed by Asperti and Roversi~\cite{AspertiRoversi:TOCL02} after Girard's 
treatment of the purely linear case. This would also allow unrestricted erasing of
processes, leaving marked boxes only for duplication. 
Second, individuate a richer language of processes, still amenable to the soft (or light)
treatment. Section~\ref{sect:eshopi} suggests a possible direction, but many others are possible.
Third, the very interesting problem of studying the notion of complexity
class in the process realm. 

In the paper, we have proved polynomial bounds for
\shopi, obtained from the the Higher-Order $\pi$-calculus by imposing
constraints inspired by Soft Linear Logic.
We have then considered an extension of  \shopi, taking into account
features specific to processes, notably the existence of channels: 
in process calculi
a
reduction step does not need to be anonymous, as in the $\lambda$-calculus, but
may result from an interaction along a channel. 
An objective of the extension was to accept processes that are
programmed  to have unboundedly many external actions (i.e., 
interactions with their
environment) but that  remain polynomial on the internal work  
 performed between any two  external activities. 
Our definition of the extended  class, \eshopi{\inpchans}, relies on the 
notion of input channel~--- a channel that is used in a process only
in input. This allows us to have more flexibility in the permitted
forms of copying. 
We have proposed \eshopi{\inpchans}\ because this class seems mathematically
simple and practically interesting. These claims, however, need 
to be sustained by
more evidence.
Furher, other refinements  of  \shopi\ are possible. 
Again, more experimentation with examples is
needed to understand where to focus attention.

Another question related to the interplay between internal and external
actions of processes is whether the polynomial bounds on internal
actions change when external actions are performed. 

Summarizing, we started with a question (``Can ICC be applied to process algebras?'')
and ended up with a positive answer and many more different questions. But this is a feature,
and not a bug.

\bibliographystyle{eptcs}
{\small \bibliography{main}}
\end{document}